\documentclass[a4paper, 12pt]{article}

\usepackage[height=23.5cm, width=16.5cm, hmarginratio={1:1}]{geometry}

\usepackage{cite}
\usepackage{amssymb}
\usepackage{amsmath, amsthm}
\allowdisplaybreaks[1]
\usepackage{graphicx,color}
\usepackage{wrapfig,framed,bbm}
\usepackage{hyperref}
\usepackage{indentfirst}
\usepackage{mathrsfs}
\usepackage{geometry}
\usepackage{appendix}
\usepackage{caption}
\usepackage{subcaption}
\usepackage{booktabs}
\usepackage[figuresright]{rotating}
\usepackage{bm}

\usepackage{graphicx}
\usepackage{amssymb,color}
\usepackage{cancel}

\usepackage{caption}
\usepackage{subcaption}
\usepackage{tensor}
\usepackage{mathtools}
\usepackage{mhchem}
\usepackage{appendix}

\usepackage{microtype}

\usepackage{hyperref}

\newcommand{\beq}{\begin{equation}}
\newcommand{\eeq}{\end{equation}}

\newcommand{\Tr}{{\rm Tr}}
\newcommand{\bt}{{\bf t}}

\def\={\; = \;}
\def\+{\; + \;}
\def\:={\, := \, }

\begin{document}
	
\title{Wave functions and $k$-point functions for the AKNS hierarchy}
\author{Ang Fu}
\date{}
\maketitle
		
\theoremstyle{plain}
\numberwithin{equation}{section}
\newtheorem{theorem}{Theorem}[section]
\newtheorem{lemma}{Lemma}[section]
\newtheorem{Definition}{Definition}[section]
\newtheorem{remark}{Remark}[section]	
\newtheorem{pro}{Proposition}[section]
\newtheorem{cor}{Corollary}[section]
		
\begin{abstract}
For an arbitrary solution to the AKNS hierarchy, the logarithmic derivatives of the tau-function of the solution can be computed by the matrix-resolvent method~\cite{D20, FY22}. In this paper, we introduce a pair of wave functions of the solution and we use them to express the corresponding matrix resolvent. Based on this, we derive a new formula for the $k$-point correlation function of the AKNS hierarchy expressed in terms of wave functions. As an application, we show that the tau-function of an arbitrary solution to the AKNS hierarchy is a KP tau-function.

\end{abstract}

	
%   \setcounter{tocdepth}{1}
%\tableofcontents		
		
\section{Introduction}\label{sec1}
The matrix-resolvent (MR) method, introduced and developed in~\cite{BDY16, BDY21, CY22, D20, DS85, DVY22, DY17, DYZ21, FY22, FYZ23, Y20, YZhou}, provides an efficient way for computing logarithmic derivatives of tau-functions for integrable systems. The MR method is also effective in exploring relationships between different integrable systems (see, e.g.,~\cite{FY22, FYZ23, YZhou}).

In~\cite{DYZ21}, for an arbitrary solution to the KdV hierarchy, an explicit formula for the generating series of $k$th order $(k \ge 2)$ logarithmic derivatives  of the corresponding tau-function was obtained, which is explicitly a pair of wave functions  of the solution; 
as applications, computations on the Witten–Kontsevich tau-function, the generalized
Brézin–Gross–Witten (BGW) tau-function and a modular deformation of the
generalized BGW tau-function were given there. 
This development above was subsequently extended~\cite{Y20} to the Toda lattice hierarchy; as applications, computations on the 
Gaussian Unitary Ensemble (GUE) correlators and Gromov–Witten invariants of the Riemann sphere were  given there. 
The similar generalization was also given~\cite{FLY24} for the Volterra lattice hierarchy.

In~\cite{FY22}, the MR method was extended to the AKNS hierarchy (also known as the nonlinear Schr\"{o}dinger (NLS) hierarchy or the $1$-constrained KP hierarchy; cf.~\cite{AC91, AKNS74, C92, CL91, D03, L99, T04, ZS72}). Based on this method, we gave~\cite{FY22} a detailed proof of a theorem of Carlet, Dubrovin and Zhang~\cite{CDZ04} regarding the relationship between the (extneded) Toda lattice hierarchy and the (entended) AKNS hierarchy. 
From the well-known Lax representation of the AKNS~\cite{AC91, AKNS74, FY22, ZS72}, the MR method~\cite{FY22}, and the similarity to Toda lattice hierarchy, it is natural to expect that the AKNS hierarchy admits a pair of wave functions
$\psi_A(\bt;\xi;\epsilon)$ and $\psi_B(\bt;\xi;\epsilon)$, and  these functions should yield novel explicit formulae for the generating series of $k$th order $(k \ge 2)$ logarithmic derivatives  of the tau-function for the AKNS hierarchy. 

Let us now give the definition of the formal wave functions for the AKNS hierarchy. Let $V$ be the ring of functions of $X$ closed under $\partial_X$, and let $\widetilde{V}$ be a ring with $V\subseteq \partial_X(\widetilde{V})\subseteq\widetilde{V}$.  Let $(q(\bt;\epsilon),r(\bt;\epsilon)),\bt=(t_0=X,t_1,t_2,\dots)$ be the unique solution in $V((\epsilon))[[\bt_{>0}]]^2$ to the AKNS hierarchy, satisfying the initial condition $q(\bt;\epsilon)\vert_{\bt_{>0}=0}=q(X;\epsilon),r(\bt;\epsilon)\vert_{\bt_{>0}=0}=r(X;\epsilon)$ with $q(X;\epsilon),r(X;\epsilon)\in V((\epsilon))$.
Let $R({\bf t};\xi;\epsilon)$ be the matrix resolvent~\cite{FY22} for the AKNS hierarchy (see also in Section~\ref{sec2}).
\begin{Definition}\label{t-wavefun-def}
	An element
	\begin{align}
		\psi_A(\bt;\xi;\epsilon)=\phi_{A}(\bt,\xi;\epsilon)e^{\epsilon^{-1}\sum_{k\ge 0}2^{k}t_{k}\xi^{k+1}},\label{waveta}
	\end{align}
	where $\phi_{A}(\bt,\xi;\epsilon)=\sum_{k=0}^{+\infty}\phi_{A,k}(\bt;\epsilon)\xi^{-k}\in \widetilde{V}((\epsilon))[[\xi^{-1}]]$ and $\phi_{A,0}=1$, 
	is called the formal wave functions of type A associated with $(q(\bt;\epsilon),r(\bt;\epsilon))$ if
	\begin{align}
		\frac{\partial \Psi_{A}(\bt;\xi;\epsilon)}{\partial t_{k}}=\epsilon^{-1}2^{k}\left(\left(\xi^{k+1}R({\bf t};\xi;\epsilon)\right)_{+}-\xi^{k+1}I\right) \Psi_{A}(\bt;\xi;\epsilon),\quad k\ge 0,\label{spectraleqtimeA}
	\end{align}
	where $\Psi_{A}(\bt;\xi;\epsilon)=\begin{pmatrix}
		\psi_A(\bt;\xi;\epsilon),
		\frac{\epsilon\psi_{A,X}(\bt;\xi;\epsilon)-\xi\psi_A(\bt;\xi;\epsilon)}{q(\bt;\epsilon)}
	\end{pmatrix}^{T}$ and $\psi_{A,X}$ denotes $\frac{\partial \psi_{A}}{\partial X}$.
	An element
	\begin{align}
		&\psi_B(\bt;\xi;\epsilon)=\phi_{B}(\bt,\xi;\epsilon)q(\bt;\epsilon)e^{\epsilon^{-1}\sum_{k\ge 0}-2^{k}t_{k}\xi^{k+1}},\label{wavetb}
	\end{align}
	where $\phi_{B}(\bt,\xi;\epsilon)=\sum_{k=0}^{+\infty}\phi_{B,k}(\bt;\epsilon)\xi^{-k}\in \widetilde{V}((\epsilon))[[\xi^{-1}]]$ and $\phi_{B,0}=1$,
	is called the formal wave functions of type B associated with $(q(\bt;\epsilon),r(\bt;\epsilon))$ if
	\begin{align}
		\frac{\partial \Psi_{B}(\bt;\xi;\epsilon)}{\partial t_{k}}=\epsilon^{-1}2^{k}\left(\left(\xi^{k+1}R({\bf t};\xi;\epsilon)\right)_{+}-\xi^{k+1}I\right) \Psi_{B}(\bt;\xi;\epsilon),\quad k\ge 0,\label{spectraleqtimeB}
	\end{align}
	where $\Psi_{B}(\bt;\xi;\epsilon)=\begin{pmatrix}
		\psi_B(\bt;\xi;\epsilon),
		\frac{\epsilon\psi_{B,X}(\bt;\xi;\epsilon)-\xi\psi_B(\bt;\xi;\epsilon)}{q(\bt;\epsilon)}
	\end{pmatrix}^{T}$ and $\psi_{B,X}$ denotes $\frac{\partial \psi_{B}}{\partial X}$.
\end{Definition}
The existence of the two wave functions $\psi_A$ and $\psi_B$ will be established in Section~\ref{sec3}.  

We define
\begin{align}
	d(\bt;\xi;\epsilon)=\epsilon\frac{\psi_A(\bt;\xi;\epsilon)\psi_{B,X}(\bt;\xi;\epsilon)-\psi_{A,X}(\bt;\xi;\epsilon)\psi_B(\bt;\xi;\epsilon)}{q(\bt;\epsilon)}.\label{definedtime}
\end{align}
\begin{Definition}
	We say $\psi_A,\psi_B$ form a pair if 
	\begin{align}\label{paircond}
		d(\bt;\xi;\epsilon)=-2\xi.
	\end{align}
\end{Definition}
We will prove that there exists a pair of wave functions $\psi_A, \psi_B$ associated with $(q(\mathbf{t};\epsilon), r(\mathbf{t};\epsilon))$ for the AKNS hierarchy in Section~\ref{sec3}.  
Similar to that in~\cite{D03, DYZ21, FLY24, Y20},  the relationship between the pair of wave functions $\psi_A,\psi_B$
​and the matrix resolvent for the AKNS hierarchy will be established in Section~\ref{sec4} and we introduce 
\begin{align}\label{defineD}
	D(\bt;\xi,\nu;\epsilon):=&\frac{\epsilon\left(\psi_A(\bt;\xi;\epsilon)\psi_{B,X}(\bt;\nu;\epsilon)-\psi_{A,X}(\bt;\xi;\epsilon)\psi_B(\bt;\nu;\epsilon)\right)}{(\xi-\nu)q(\bt;\epsilon)}+\frac{\psi_A(\bt;\xi;\epsilon)\psi_B(\bt;\nu;\epsilon)}{q(\bt;\epsilon)}.
\end{align}
Then we have the following main theorem.

\begin{theorem}\label{main1time} 
	Let $k\geq 2$ be an integer. 
	The generating series of $k$-point correlation functions of the solution $(q(\bt;\epsilon),r(\bt;\epsilon))$ is given by
	\begin{align}
		&\sum_{i_1,\dots,i_k\ge 0}\Omega_{i_{1},i_{2},\dots,i_{k}}({\bf t};\epsilon)\prod_{j=1}^{k}\frac{1}{2^{i_j}\xi_{j}^{i_j+2}}\nonumber\\
		=&  
		-\frac{1}{\prod_{j=1}^k \xi_{j}}  
		\sum_{\sigma \in S_k/C_k} \prod_{j=1}^k D(\bt;\xi_{\sigma(j+1)},\xi_{\sigma(j)};\epsilon) 
		\,-\, \frac{4\delta_{k,2}}{(\xi_1-\xi_2)^2},  \label{mainidentity}
	\end{align}
	where $\Omega_{i_{1},i_{2},\dots,i_{k}}({\bf t};\epsilon)$ are the $k$-point correlation functions (defined in Section~\ref{sec2}) of the solution $(q(\bt;\epsilon),r(\bt;\epsilon))$ for the AKNS hierarchy.

\end{theorem}
We note that for algebraic geometric solution an essentially equivalent form of the above Theorem~\ref{main1time} was also given in~\cite{D20}.

In~\cite{Z15}, Zhou derived  an explicit formula  for the $n$-point function for an arbitrary KP tau-function in the big cell. We now give a brief review of this formula. Let
$Z_V({\bf t})$ be an arbitrary KP tau-function in the big cell, where $V$ corresponds to an element in the big cell $Gr_0$ of Sato's Grassmannian, whose affine coordinates are $A_{i,j}$.
The associated $n$-point function for $\log Z_V$ is defined as
\begin{align}
G_n(\xi_1, \dots, \xi_n) = \sum_{k_1,\dots,k_n\ge 1} \prod_{i=1}^n \frac{1}{\xi_i^{k_i+1}} \cdot
\frac{\partial^n \log Z_V({\bf t})}{\partial t_{k_1-1} \dots \partial t_{k_n-1}}\biggr|_{{\bf t} = {\bf 0}}, \quad n\ge 1.
\end{align}
Zhou proved~\cite{Z15} that for $n=1$, $G_1(\xi) = \sum_{i,j\geq 0} A_{i,j} \xi^{-i-j-2}$,
and for $n \ge2 $,
\begin{align} \label{eqn:G-n}
	G_n(\xi_1, \ldots, \xi_n)
	= (-1)^{n-1} \sum_{\sigma \in S_{n}/C_{n}}
	\prod_{i=1}^n B(\xi_{\sigma(i)}, \xi_{\sigma(i+1)})
	-  \frac{\delta_{n,2}}{(\xi_1-\xi_2)^2},
\end{align}
where 
\begin{align}
	B(\xi_i, \xi_j) = \begin{cases}
		\frac{1}{\xi_i-\xi_j} + A(\xi_i, \xi_j),  & i \not= j, \\
		A(\xi_i, \xi_i),  & i =j, 
	\end{cases}
\end{align}
and $A(\xi, \eta)
=  \sum_{i,j \geq 0} A_{i,j}   \xi^{-j-1} \eta^{-i-1} $. 
It was shown in~\cite{YZhou} that the above formula directly implies that 
for any formal power series $F$ in $\mathbf{t}$ with associated $n$-point functions $G_n(\xi_1,\dots,\xi_n)$, 
if there exists numbers $(A_{i,j})_{i,j\geq 0}$ such that for all $n \geq 2$, $G_n$ is given by~\eqref{eqn:G-n}, then $Z = e^F$ is a KP tau-function.
In~\cite{ABDKS25}, both~\eqref{eqn:G-n} and the inverse statement are called Zhou's theorem.
Using the MR method for the Toda lattice hierarchy, Yang and Zhou~\cite{YZhou} proved
%Following this result and using the matrix resolvent (MR) method for the extended Toda lattice hierarchy~\cite{YZhou}, Yang and Zhou~\cite{YZhou} further proved 
that the tau-function of an arbitrary solution to the Toda lattice hierarchy is a KP tau-function.  Similiar to~\cite{YZhou}, we will prove the following theorem.

\begin{theorem}\label{kptau}
	The tau-function $\tau({\bf t};\epsilon)$ of an arbitrary solution to the AKNS hierarchy in the sense of~\cite{D20,FY22} is a tau-function of the KP hierarchy.
\end{theorem}
We note that this theorem also follows from~\cite{CDZ04, FY22, YZhou} for a wide class of solutions.

%We will give a new proof for the theorem in Section~\ref{sec4}.

The rest of this paper is organized as follows. In Section~\ref{sec2}, we review the MR method for
the study of tau-functions for the AKNS hierarchy. In Section~\ref{sec3}, we construct a pair
of wave functions for the AKNS hierarchy. In Section~\ref{sec4}, we give explicit formulas for
generating series of the $k$-point correlation functions for the AKNS hierarchy in terms
of the wave functions and we prove Theorem~\ref{main1time} and Theorem~\ref{kptau}.

\section{Review of the MR method to tau-functions for the AKNS hierarchy}\label{sec2}
Let us give a brief review of the MR method of computing logarithmic derivatives of tau-functions for the AKNS hierarchy~\cite{D20,FY22}. Let
$\mathcal{A}:=\mathbb{C}[q_{kX},r_{kX}|k\ge 0]$ be the polynomial ring. Here $q_{kX}=\partial_{X}^{k}(q),r_{kX}=\partial_{X}^{k}(r), k\ge 0$. 
Let $\mathcal{L}(\xi)$ be the matrix Lax operator (cf.~e.g.~\cite{AC91,D03, L99, NMPZ84,T04}) 
\begin{align}
	\mathcal{L}(\xi)=\begin{pmatrix}
		\epsilon\partial_{X}&0\\
		0&\epsilon\partial_{X}
	\end{pmatrix}+\begin{pmatrix}
		-\xi & -q\\
		r & \xi\\
	\end{pmatrix}.\label{laxnls}
\end{align}
There exists a unique series $R(\xi)$ satisfying
\begin{align}
	&	R(\xi)-\begin{pmatrix}
		2&0\\
		0&0\\
	\end{pmatrix}
	\in\mathrm{Mat} \left(2, \mathcal{A}[\epsilon]\left[\left[\xi^{-1} \right] \right]\xi^{-1}\right)\label{rnlsdef},\\
	&   \left[\mathcal{L}(\xi), R(\xi)\right]=0,\quad    \Tr R(\xi)=2, \quad \mathrm{det}\ R(\xi)=0.\label{22}	
\end{align}
The unique $R(\xi)$ is called the basic matrix resolvent of $\mathcal{L}(\xi)$. Write
\begin{align}
	&   R(\xi)=\begin{pmatrix}
			2+\sum_{j\ge 0} \frac{A_j}{\xi^{j+1}}&\sum_{j\ge 0} \frac{B_j}{\xi^{j+1}}\\
			\sum_{j\ge 0} \frac{C_j}{\xi^{j+1}}&-\sum_{j\ge 0} \frac{A_j}{\xi^{j+1}}
		\end{pmatrix},
\end{align}
where $A_j,B_j,C_j\in\mathcal{A}[\epsilon]$ and are uniquely determined by the following recurrence relation
\begin{align}
	&	\epsilon\partial_{X}(A_{j})-rB_j-qC_j=0,\quad j\ge 0, \label{abc1}\\
	&	\epsilon\partial_{X}(B_{j})+2qA_j-2B_{j+1}+2q\delta_{-1,j}=0,\quad j\ge -1,\label{abc2}\\
	&	\epsilon\partial_{X}(C_{j})+2rA_j+2C_{j+1}+2r\delta_{-1,j}=0,\quad j\ge -1, \label{abc3}\\
	&      A_k=-\frac{1}{2}\sum_{\substack{ i+j=k-1\\ i,j\ge -1}} (A_iA_j+B_iC_j), \qquad k\geq 0.\label{23}
\end{align}
Here $\delta_{i,j}$ denotes the Kronecker delta and $A_{-1}=B_{-1}=C_{-1}:=0$. 
Recall that the AKNS hierarchy was defined via
\begin{align}\label{Tnlshierarchy}
	\frac{\partial \mathcal{L}(\xi)}{\partial t_{j}}=2^{j}\epsilon^{-1} \Bigl[V_{j}(\xi),\mathcal{L}(\xi)\Bigr],\quad j\ge  0,
\end{align}
where $V_{j}(\xi)=\left(\xi^{j+1}R(\xi)\right)_{+}$. 
Then the following equation holds true:
	\begin{gather}
		\epsilon\nabla(\nu)\left(R(\xi)\right)=\frac{\left[R( \nu), R(\xi)\right]}{\nu-\xi}+\left[-\frac{1}{\nu}\begin{pmatrix}
			2&0\\
			0&0
		\end{pmatrix}, R(\xi)\right]\label{derivition},
	\end{gather}
	where $\nabla(\nu):=\sum_{j\ge 0} \frac{1}{2^j\nu^{j+2}}\frac{\partial}{\partial t_{j}}$.
\textit{The tau-structure}~\cite{CDZ04,DZ,DZ04, FY22} for the AKNS hierarchy~\eqref{Tnlshierarchy} are defined via the generating series
\begin{align}
	\sum_{i, j\ge 0}\frac{1}{\xi^{i+2}\nu^{j+2}}\frac{\Omega_{i,j}}{2^{i+j}}=\frac{\Tr R(\xi)R(\nu)-4}{(\xi-\nu)^2}.   \label{tastr}
\end{align}	
Let $(q({\bf t};\epsilon),r({\bf t};\epsilon))$ be an arbitrary solution to the AKNS hierarchy~\eqref{Tnlshierarchy}. 
Then there exists a function $\tau({\bf t};\epsilon)$, such that
\begin{align}
	\Omega_{i,j}({\bf t};\epsilon)=\epsilon^2\frac{\partial^2 \log\tau({\bf t};\epsilon)}{\partial t_{i}\partial t_{j}},\quad i,j\ge  0.\label{DZtauf}
\end{align}
The function $\tau({\bf t};\epsilon)$ is called the \textit{Dubrovin-Zhang type tau-function} 
of the solution $(q,r)$ to the AKNS hierarchy. 
The function $\tau({\bf t};\epsilon)$ is determined uniquely by the solution $(q,r)$ up to 
multiplying by the exponential of a linear function.
The logarithmic derivatives $\epsilon^k\frac{\partial^k\log\tau({\bf t};\epsilon)}{\partial t_{i_1}\cdots\partial t_{i_k}}=:\Omega_{i_1,\dots,i_k}({\bf t};\epsilon), k\ge 2,\,i_1,\dots,i_k\ge 0$, are called the $k$-point correlation functions of the solution.
For any $k\ge 2$, it was proven in~\cite{D20,FY22} that the following formula holds  true:
	\begin{align}
		&\sum_{i_1,\dots,i_k\ge 0} \epsilon^k\frac{\partial^k\log\tau({\bf t};\epsilon)}{\partial t_{i_1}\cdots\partial t_{i_k}}\prod_{j=1}^{k}\frac{1}{2^{i_j}\xi_{j}^{i_j+2}} \nonumber\\
		=&-\sum_{\sigma\in S_k/ C_k}\frac{\Tr\left(R({\bf t};\xi_{\sigma(1)};\epsilon)\cdots R({\bf t};\xi_{\sigma(k)};\epsilon)\right)}{(\xi_{\sigma(1)}-\xi_{\sigma(2)})\cdots(\xi_{\sigma(k-1)}-\xi_{\sigma(k)})(\xi_{\sigma(k)}-\xi_{\sigma(1)})}-\frac{4\delta_{k,2}}{(\xi_1-\xi_2)^2}.\label{cor2,21}
	\end{align}

\section{A pair of wave functions for the AKNS hierarchy} \label{sec3}

In this section, for an arbitrary solution to the AKNS hierarchy, we define a pair
of wave functions of the solution. 

We first consider the time independent case. 
Let $q=q(X;\epsilon),r=r(X;\epsilon)$ be the element of $V((\epsilon))$. Consider the following matrix operator:
\begin{align}
	\mathcal{L}(X;\xi;\epsilon)=\epsilon\partial_X+\begin{pmatrix}
		-\xi&-q(X;\epsilon)\\
		r(X;\epsilon)&\xi
	\end{pmatrix}.
\end{align}

\begin{Definition}
An element
\begin{align}
	\psi_A(X;\xi;\epsilon)=\phi_{A}(X,\xi;\epsilon)e^{\epsilon^{-1}X\xi}\in\widetilde{V}((\epsilon))[[\xi^{-1}]]e^{\epsilon^{-1}X\xi},\label{wavea}
\end{align}
where $\phi_{A}(X,\xi;\epsilon)=\sum_{k=0}^{+\infty}\phi_{A,k}(X;\epsilon)\xi^{-k}$ and $\phi_{A,0}=1$,
is called the formal wave functions of type A associated with $(q(X;\epsilon), r(X;\epsilon))$ if
\begin{align}
	\mathcal{L}(X;\xi;\epsilon)\begin{pmatrix}
		\psi_A(X;\xi;\epsilon)\\
		\frac{\epsilon\psi_{A,X}(X;\xi;\epsilon)-\xi\psi_A(X;\xi;\epsilon)}{q(X;\epsilon)}
	\end{pmatrix}=0. \label{spectraleqA}
\end{align}
An element
\begin{align}
\psi_B(X;\xi;\epsilon)=\phi_{B}(X,\xi;\epsilon)q(X;\epsilon)e^{-\epsilon^{-1}X\xi}\in\widetilde{V}((\epsilon))[[\xi^{-1}]]e^{\epsilon^{-1}X\xi} ,\label{waveb}
\end{align}
where $\phi_{B}(X,\xi;\epsilon)=\sum_{k=0}^{+\infty}\phi_{B,k}(X;\epsilon)\xi^{-k}$ and $\phi_{B,0}=1$, is called the formal wave functions of type B associated with $(q(X;\epsilon), r(X;\epsilon))$ if
\begin{align}
\mathcal{L}(X;\xi;\epsilon)\begin{pmatrix}
	\psi_B(X;\xi;\epsilon)\\
	\frac{\epsilon\psi_{B,X}(X;\xi;\epsilon)-\xi\psi_B(X;\xi;\epsilon)}{q(X;\epsilon)}
\end{pmatrix}=0. \label{spectraleqB}
\end{align}
\end{Definition}
Let us give a proof of the existence of wave functions of types A and B associated with $(q(X;\epsilon), r(X;\epsilon))$. 
Through equations \eqref{spectraleqA} and \eqref{spectraleqB}, $\psi_A$ and $\psi_B$ emerge as linearly independent solutions of the second-order differential equation:
\begin{align}
	\epsilon^2\psi_{XX}-\epsilon^2\frac{q_{X}}{q}\psi_{X}+\left(\epsilon\frac{q_{X}}{q}\xi+qr-\xi^2\right)\psi=0.\label{waveequation}
\end{align}
Let $f(X;\epsilon) := \dfrac{q_{X}(X;\epsilon)}{q(X;\epsilon)}$ and $g(X;\epsilon) := q(X;\epsilon)r(X;\epsilon)$ for notational simplicity. 
Write
\begin{align}
	&\psi_A(X;\xi;\epsilon)=e^{\epsilon^{-1}\partial_X^{-1}x(X;\xi;\epsilon)}e^{\epsilon^{-1}X\xi},\quad x(X;\xi;\epsilon)=\sum_{k\ge 1}^{}\frac{x_k(X;\epsilon)}{\xi^k},\\
	&\psi_B(X;\xi;\epsilon)=e^{\epsilon^{-1}\partial_X^{-1}y(X;\xi;\epsilon)}q(X;\epsilon)e^{-\epsilon^{-1}X\xi},\quad y(X;\xi;\epsilon)=\sum_{k\ge 1}^{}\frac{y_k(X;\epsilon)}{\xi^k}.
\end{align}
Then, the equation (\ref{waveequation}) for $\psi=\psi_A$ and for $\psi=\psi_B$ recast into the following equations:
\begin{align}
	&\epsilon x_{X}(X;\xi;\epsilon)+x(X;\xi;\epsilon)^2+(2\xi-\epsilon f(X;\epsilon))x(X;\xi;\epsilon)+g(X;\epsilon)=0,\label{eq1}\\
	&\epsilon y_{X}(X;\xi;\epsilon)+y(X;\xi;\epsilon)^2+(\epsilon f(X;\epsilon)-2\xi)y(X,\xi)+ \epsilon^{2}f_{X}(X;\epsilon)+g(X;\epsilon)=0.\label{eq2}
\end{align}
These yield recursive relations for the coefficients:
\begin{align}
	&x_{k+1}=\frac{1}{2}(-\epsilon x_{k,X}-\sum_{m=1}^{k-1}x_mx_{k-m}+\epsilon fx_k)-\frac{1}{2}g\delta_{k,0},\label{solution1}\\
	&y_{k+1}=\frac{1}{2}(\epsilon y_{k,X}+\sum_{m=1}^{k-1}x_mx_{k-m}+\epsilon fy_k)+\frac{1}{2}(\epsilon^2f_{X}+g)\delta_{k,0},
\end{align}
for $k\ge 0$.  
From these recursions, it easily follows that $x_k,y_k\in V((\epsilon))$, $k\geq 0$.  This proves the existence of wave functions of type A and type B meeting the definitions \eqref{spectraleqA}--\eqref{spectraleqB}.
Clearly, $\psi_A$ and $\psi_B$ are unique up to multiplying by 
$G(\xi;\epsilon)\in 1+ \mathbb{C}((\epsilon))[[\xi^{-1}]]\xi^{-1}$ and 
$E(\xi;\epsilon) \in 1+ \mathbb{C}((\epsilon))[[\xi^{-1}]]\xi^{-1}$.
Define
\begin{align}
	d(X;\xi;\epsilon)=\epsilon\frac{\psi_A(X;\xi;\epsilon)\psi_{B,X}(X;\xi;\epsilon)-\psi_{A,X}(X;\xi;\epsilon)\psi_B(X;\xi;\epsilon)}{q(X;\epsilon)}.\label{defined}
\end{align}
We call $\psi_A,\psi_B$ form \textit{a pair} if $d(X;\xi;\epsilon)=-2\xi$.
Using~\eqref{wavea} and~\eqref{waveb}, 
we find that the $d(X;\xi;\epsilon)$ defined in \eqref{defined} must have the form $d(X;\xi;\epsilon)=-2\xi e^{\sum_{k\ge 1}^{}d_k(X;\epsilon)\xi^{-k}}$.
Then by using \eqref{waveequation}, one can easily derive 
\begin{align}
	\frac{\partial (d(X;\xi;\epsilon))}{\partial X}=0.
\end{align}
It follows that all $d_{k}(X;\epsilon),k\ge 0$, are constants. Therefore, for any fixed choice of $\psi_A,$ we can suitably choose the factor $E(\xi;\epsilon)$ for $\psi_B$ such that $\psi_A,\psi_B$ form a pair. This proves the exitence of pair of wave functions associated to $(q(X;\epsilon),r(X;\epsilon))$.

We proceed with the time dependent case.  Recall that the definition of the two wave functions of type A $\psi_A(\bt;\xi;\epsilon)$ and type B $\psi_B(\bt;\xi;\epsilon)$ associated with $(q(\bt;\epsilon),r(\bt;\epsilon))$ in Section~\ref{sec1}.

\begin{lemma}\label{exist-wave}
The functions $\psi_A(\bt;\xi;\epsilon)$ and $\psi_B(\bt;\xi;\epsilon)$, defined via~Definition~\ref{t-wavefun-def}, exist.
\end{lemma}
\begin{proof}
We prove the lemma by relating $\psi_A$ and $\psi_B$ to the wave function and dual wave function of the AKNS hierarchy.
 Let 
$\phi=\sum_{i=0}^{+\infty}w_{i}(\epsilon\partial_{X})^{-i},w_{0}=1,$ be the dressing operator of the scaled AKNS Lax operator $L=\epsilon\partial_{X}+q\circ(\epsilon\partial_{X})^{-1}\circ r$,
so that 
$
	L=\phi\circ (\epsilon\partial_{X})\circ \phi^{-1}.
$
Setting $\lambda=2\xi$, the wave function $\psi=\psi({\bt};\lambda;\epsilon)$ and the dual wave function $\psi^{*}=\psi^{*}({\bt};\lambda;\epsilon)$ are defined by
\begin{align}
	\psi=\phi(e^{\epsilon^{-1}\sum_{k=0}^{+\infty}t_{k}\lambda^{k+1}}),\quad 	\psi^{*}=(\phi^{-1})^{*}(e^{-\epsilon^{-1}\sum_{k=0}^{+\infty}t_{k}\lambda^{k+1}}),
\end{align}
where $(\phi^{-1})^{*}$ denotes the dual operator of $\phi^{-1}$.
They satisfy
\begin{align}
	&L\psi=\lambda \psi,\label{scl-wave-1'}\\
	& \frac{\partial \psi}{\partial t_{k}}=\epsilon^{-1}\left(L^{k+1}\right)_{+}\psi,\label{scl-wave-1}\\
	&L^{*}\psi^{*}=\lambda \psi^{*},\label{scl-wave-2'}\\
	& \frac{\partial \psi^{*}}{\partial t_{k}}=-\epsilon^{-1}\left((L^{*})^{k+1}\right)_{+}\psi^{*},\label{scl-wave-2}
\end{align}
where $k\ge 0$, $L^{*}=-\epsilon\partial_{X}-r\circ(\epsilon\partial_{X})^{-1}\circ q$ and $\lambda=2\xi$. It has proved in~\cite{C92}
\begin{align}
	&\left(L^{k+1}\right)_{+}(q)=\epsilon\partial_{X}\left(\left(L^{k}\right)_{+}(q)\right)+2q(\epsilon \partial_{X})^{-1}\left(r\left(L^{k}\right)_{+}(q)-q\left(\left(L^{*}\right)^{k}\right)_{+}(r)\right),\\
	&\left((L^{*})^{k+1}\right)_{+}(r)=-\epsilon\partial_{X}\left(\left((L^{*})^{k}\right)_{+}(r)\right)+2r(\epsilon \partial_{X})^{-1}\left(r\left(L^{k}\right)_{+}(q)-q\left(\left(L^{*}\right)^{k}\right)_{+}(r)\right).
\end{align}
Then, by using~\eqref{abc1}--\eqref{23} and $A_{0}=0,B_{0}=q,C_{0}=-r$,
we derive the relations
\begin{align}
	&(\epsilon\partial_{X})^{-1}\left(r\left(L^{k}\right)_{+}(q)-q\left(\left(L^{*}\right)^{k}\right)_{+}(r)\right)={\rm res}_{\epsilon\partial}L^{k}=-{\rm res}_{\epsilon\partial}(L^*)^{k}=2^{k}A_{k},\\
	&\left(L^{k}\right)_{+}(q)=2^{k}B_{k},\quad -\left(\left(L^{*}\right)^{k}\right)_{+}(r)=2^{k}C_{k}.
\end{align}
By using these relations and the following formula in~\cite{BLX96}
\begin{align}
	&\left(L^{k+1}\right)_{+}=\sum_{l=0}^{k}{\rm res}_{\epsilon\partial}L^{l}\circ L^{k-l}+L^{k+1}-\sum_{l=0}^{k}\left(L^{l}\right)_{+}(q)\circ(\epsilon\partial_{X})^{-1}\circ r\circ L^{k-l},\\
	&\left((L^{*})^{k+1}\right)_{+}=-\sum_{l=0}^{k}{\rm res}_{\epsilon\partial}(L^{*})^{l}\circ (L^{*})^{k-l}+(L^{*})^{k+1}+\sum_{l=0}^{k}\left((L^{*})^{l}\right)_{+}(r)\circ(\epsilon\partial_{X})^{-1}\circ q\circ (L^{*})^{k-l},
\end{align}
we can rewrite the equation~\eqref{scl-wave-1} and~\eqref{scl-wave-2} by using the elements of the matrix-resolvent $R(\lambda)$ 
\begin{align}
	 &\frac{\partial \psi}{\partial t_{k}}=\epsilon^{-1}\left(\lambda^{k+1}+\sum_{l=0}^{k}2^{l}A_{l}\lambda^{k-l}\right)\psi+\epsilon^{-1}\sum_{l=0}^{k}2^{l}B_{l}\lambda^{k-l}\frac{\epsilon\psi_{X}-\lambda \psi}{q},\label{1}\\
	 &\frac{\partial \psi^{*}}{\partial t_{k}}=-\epsilon^{-1}\sum_{l=0}^{k}2^{l}C_{l}\lambda^{k-l}\frac{\epsilon\psi^{*}_{X}+\lambda \psi^{*}}{r}-\epsilon^{-1}\left(\lambda^{k+1}+\sum_{l=0}^{k}2^{l}A_{l}\lambda^{k-l}\right)\psi^{*}.\label{2}
\end{align}
By using~\eqref{1}--\eqref{2}, \eqref{scl-wave-1'}, \eqref{scl-wave-2'} and~\eqref{abc1}--\eqref{abc3}, we obtain  
%\clr{We have}
\begin{align}
	&\frac{\partial}{\partial t_{k}}\left(\frac{\epsilon\psi_{X}-\lambda \psi}{q}\right)=\epsilon^{-1}\sum_{l=0}^{k}2^{l}C_{l}\lambda^{k-l} \psi-\epsilon^{-1}\sum_{l=0}^{k}2^{l}A_{l}\lambda^{k-l}\frac{\epsilon\psi_{X}-\lambda \psi}{q},\\
	&\frac{\partial}{\partial t_{k}}\left(\frac{\epsilon\psi^{*}_{X}+\lambda \psi^{*}}{r}\right)=\epsilon^{-1}\sum_{l=0}^{k}2^{l}A_{l}\lambda^{k-l}\frac{\epsilon\psi^{*}_{X}+\lambda \psi^{*}}{r}-\epsilon^{-1}\sum_{l=0}^{k}2^{l}B_{l}\lambda^{k-l}\psi^{*}.
\end{align}
We define
\begin{align}
	\psi_{A}:=\psi e^{-\alpha({\bt },\xi)},\quad \psi_{B}:=-\lambda\frac{\epsilon\psi^{*}_{X}+\lambda \psi^{*}}{ r}e^{\alpha({\bt },\xi)},\quad \alpha({\bt },\xi):=\epsilon^{-1}\sum_{k\ge 0}2^{k}t_{k}\xi^{k+1}.
\end{align}
By a direct computation, we have
\begin{align}
	\frac{\epsilon\psi_{A,X}-\xi\psi_A}{q}=\frac{\epsilon\psi_{X}-\lambda \psi}{q}e^{-\alpha({\bt },\xi)},\quad \frac{\epsilon\psi_{B,X}-\xi\psi_B}{q}=\lambda \psi^{*} e^{\alpha({\bt },\xi)}.\label{6}
\end{align}
From equations~\eqref{1}--\eqref{6}, we conclude that $\psi_A$ and $\psi_B$ satisfy all the requirements in Definition~\ref{t-wavefun-def}.
The lemma is proved.
\end{proof}

We define
\begin{align}
	m(\bt;\xi,\nu;\epsilon):=\epsilon^{-1}\left(\frac{R(\bt;\nu;\epsilon)}{\nu-\xi}+Q(\nu)\right).
\end{align}
We know from~\cite{FY22} that the matrix function $\Psi(\bt;\xi;\epsilon)=\left(\Psi_{A}(\bt;\xi;\epsilon),\Psi_{B}(\bt;\xi;\epsilon)\right)$ satisfies
\begin{align}
	\nabla(\nu)\left(\Psi(\bt;\xi;\epsilon)\right)=\left(m(\bt;\xi,\nu;\epsilon)-\epsilon^{-1}\frac{\xi}{\nu(\nu-\xi)}\right)\Psi(\bt;\xi;\epsilon).\label{wave2}
\end{align}
\begin{lemma}\label{lemmazero}
The following formula holds true:
\begin{align}
	\nabla(\nu) \Bigl(d(\bt;\xi;\epsilon)\Bigr)=0.
\end{align}	
\end{lemma}
\begin{proof}
Recalling the definition~\eqref{definedtime} for $d(\bt;\xi;\epsilon)$ and using~\eqref{wave2}, we find
\begin{align}
\nabla(\nu) \Bigl(d(\bt;\xi;\epsilon)\Bigr)&=\nabla(\nu)\Bigl({\rm det\,}\left(\Psi(\bt;\xi;\epsilon)\right)\Bigr)\nonumber\\
&={\rm tr\,}\Bigl((\Psi(\bt;\xi;\epsilon))^*\nabla(\nu)(\Psi(\bt;\xi;\epsilon))\Bigr) \nonumber\\
& = \left({\rm tr\,}(m(\bt;\xi,\nu;\epsilon))-\epsilon^{-1}\frac{2\xi}{\nu(\nu-\xi)}\right)d(\bt;\xi;\epsilon)=0,
\end{align}
where $(\Psi(\bt;\xi;\epsilon))^*$ means  the  adjoint matrix of $\Psi(\bt;\xi;\epsilon)$.
The lemma is proved.
\end{proof}
The next lemma shows the exitence of a pair.
\begin{lemma}\label{pair-wave}
	There exist a pair of wave functions $\psi_A,\psi_B$ associate to  $(q(\bt;\epsilon),r(\bt;\epsilon))$. Moreover, the freedom of the pair is characterized by a factor $G(\xi;\epsilon)$ via
	\begin{align}
		&\psi_A(\bt;\xi;\epsilon)\mapsto G(\xi;\epsilon)\psi_A(\bt;\xi;\epsilon),\quad 	\psi_B(\bt;\xi;\epsilon)\mapsto \frac{1}{G(\xi;\epsilon)}\psi_B(\bt;\xi;\epsilon),\label{psi12g}\\
		& G(\xi;\epsilon)=\sum_{j\ge 0}^{}G_j(\epsilon)\xi^{-j},\quad G_0=1,\quad G_j\in\mathbb{C}((\epsilon)),\quad j\ge 1.\label{constant}
	\end{align}
%	with $G_j\in\mathbb{C}((\epsilon)),j\ge 1.$
\end{lemma}
\begin{proof}
The proof is similiar to that in~\cite{FLY24,Y20} by using 	Lemma~\ref{lemmazero}, so we omite its details.
\end{proof}

\section{Explicit formulas for the $k$-point functions}\label{sec4}
In this section, we derive two new formulas for the generating series of the $k$-point correlation
functions of an arbitrary solution to the AKNS hierarchy.

Let $(q,r)=(q(\bt;\epsilon),r(\bt;\epsilon))\in V((\epsilon))[[\bt_{\ge 1}]]^2$ be the unique solution to the AKNS hierarchy 
with the initial value $(q(\bt;\epsilon),r(\bt;\epsilon))\vert_{\bt_{\ge 1}=0}=(q(X;\epsilon),r(X;\epsilon))$, and $(\psi_A,\psi_B)$ be a pair of wave functions associated to~$(q,r)$.  
\begin{pro} \label{propRP}
	The following identity holds true: 
\begin{align}\label{RP}
	R(\bt;\xi;\epsilon) = \Psi(\bt;\xi;\epsilon) \begin{pmatrix} 2 & 0 \\ 0 & 0 \end{pmatrix} 
	\Psi^{-1}(\bt;\xi;\epsilon).  
\end{align}
\end{pro}
\begin{proof} Define 
\begin{align}
M\=M(\bt;\xi;\epsilon)\:= 
\Psi(\bt;\xi;\epsilon) \begin{pmatrix} 2 & 0 \\ 0 & 0 \end{pmatrix} 
\Psi^{-1}(\bt;\xi;\epsilon).
\end{align}	  
It is easy to verify that $M$ satisfies 
\begin{align}
	\bigl[\mathcal{L},M\bigr](\Psi)= 0, \quad \det \, M=0,\quad{\rm tr\,}M=2.
\end{align}
The entries of $M$ in terms of the pair of wave functions read
\begin{align}
M=\frac{2}{d(\bf t;\xi;\epsilon)}\begin{pmatrix}
	\frac{\psi_A(\psi_{B,X}-\xi \psi_B)}{q}&-\psi_A\psi_B  \\
	 \frac{(\psi_{A,X}-\xi \psi_A)(\psi_{B,X}-\xi \psi_B)}{q^{2}}&-\frac{\psi_B(\psi_{A,X}-\xi \psi_A)}{q}
\end{pmatrix}.
\end{align}
It follows from \eqref{waveta}, \eqref{wavetb} and Lemma~\ref{pair-wave} that
\begin{align}
	M(\bt;\xi;\epsilon)-\begin{pmatrix} 2 & 0 \\ 0 & 0 \end{pmatrix} 
	\in {\rm Mat}\left(2,\widetilde V((\epsilon))[[\bt_{\ge 1},\xi^{-1}]]\xi^{-1}\right) \,.
\end{align}
The proposition then follows from the uniqueness of the matrix resolvent for the AKNS hierarchy.
\end{proof}
We are now to prove Theorem~\ref{main1time}.
\begin{proof}[Proof of Theorem~\ref{main1time}]
	It follows from~\eqref{RP} and (\ref{defineD}) that 
	\begin{align}\label{RP1}
		&R(\bt;\xi;\epsilon) \=    
		\frac{2r_1(\bt;\xi;\epsilon)^T r_2(\bt;\xi;\epsilon)}{d(\bt;\xi;\epsilon)},\quad D(\bt;\xi,\nu;\epsilon)=\frac{r_2(\bt;\nu;\epsilon) \, r_1(\bt;\xi;\epsilon)^T}{\xi-\nu},
	\end{align}
	where 
	\begin{align}
		&r_1(\bt;\xi;\epsilon):=\left(\psi_A(\bt;\xi;\epsilon), \frac{\epsilon\psi_{A,X}(\bt;\xi;\epsilon)-\xi\psi_A(\bt;\xi;\epsilon)}{q(\bt;\epsilon)}\right), \\
		&r_2(\bt;\xi;\epsilon):=\left(\frac{\epsilon\psi_{B,X}(\bt;\xi;\epsilon)-\xi\psi_B(\bt;\xi;\epsilon)}{q(\bt;\epsilon)},-\psi_B(\bt;\xi;\epsilon)\right).
	\end{align}
Then, substituting this expression into the identity~\eqref{tastr} and~\eqref{tastr}, the left of this proof is similiar to that in~\cite{FLY24, Y20}, we omite its details here. The theorem is proved.
\end{proof}
We define 
\begin{align}
	B(\bt;\xi,\nu;\epsilon):=\frac{D(\bt;\xi,\nu;\epsilon)}{\nu}\frac{e^{e^{\epsilon^{-1}\sum_{k\ge 0}2^{k}t_{k}\nu^{k+1}}}}{e^{\epsilon^{-1}\sum_{k\ge 0}2^{k}t_{k}\xi^{k+1}}}
\end{align}
Then Theorem~\ref{main1time} can be alternatively written in terms of $B(\bt;\xi,\nu;\epsilon)$ as shown in the following proposition.
\begin{pro}\label{Beq}
	Fix $k\geq 2$ being an integer. 
The generating series of $k$-point correlation functions of the solution $(q(\bt;\epsilon),r(\bt;\epsilon))$ has the following expression:
\begin{align}
	&\sum_{i_1,\dots,i_k\ge 0}\Omega_{i_{1},i_{2},\dots,i_{k}}({\bf t};\epsilon)\prod_{j=1}^{k}\frac{1}{2^{i_j}\xi_{j}^{i_j+2}}
	=  -\sum_{\pi \in \mathcal{S}_k/C_k} \prod_{j=1}^k B(\bt;\xi_{\sigma(j+1)},\xi_{\sigma(j)};\epsilon) 
	\,-\, \frac{4\delta_{k,2}}{(\xi_1-\xi_2)^2}  \,.  \label{mainBidentity}
\end{align}
\end{pro}

In terms of the function $\phi_A,\phi_B$, the pair condition~\eqref{paircond} reads
\begin{align}\label{paircondn}
\phi_A(\bt;\xi;\epsilon) \phi_{B,X}(\bt;\xi;\epsilon)-\phi_{A,X}(\bt;\xi;\epsilon) \phi_B(\bt;\xi;\epsilon)+\left(\frac{q_{X}}{q}-2\xi\right)\phi_A(\bt;\xi;\epsilon) \phi_B(\bt;\xi;\epsilon)=-2\xi,
\end{align}
and the function $B(\bt;\xi,\nu;\epsilon)$ reads
\begin{align}
&B(\bt;\xi,\nu;\epsilon)\nonumber\\
=&\frac{\phi_A(\bt;\xi;\epsilon) \phi_{B,X}(\bt;\nu;\epsilon)-\phi_{A,X}(\bt;\xi;\epsilon) \phi_B(\bt;\nu;\epsilon)+\left(\frac{q_{X}}{q}-2\nu\right)\phi_A(\bt;\xi;\epsilon) \phi_B(\bt;\nu;\epsilon)}{\nu(\xi-\nu)}.\label{newformB}
\end{align}
\begin{lemma}\label{lemmaB}
The function $B(\bt;\xi,\nu;\epsilon)$ admits the following expansion:
\begin{align}
B(\bt;\xi,\nu;\epsilon)=\frac{-2}{\xi-\nu}+\sum_{i,j\ge 0}\frac{A_{i,j}({\bf t};\epsilon)}{\xi^{i+1}\nu^{j+1}}.
\end{align}
\end{lemma}
\begin{proof}
	Recall the following identity:
	\begin{align}
		\phi_A(\xi)=	\phi_A(\nu)+\phi_{A,X}(\nu)(\xi-\nu)+(\xi-\nu)^{2}\partial_{\nu}\left(\frac{\phi_A(\xi)-\phi_A(\nu)}{\xi-\nu}\right),
	\end{align}
	where we omite the arguments $\bt,\epsilon$ from $\phi_A(\bt;\xi;\epsilon)$. Substituting this identity in~\eqref{newformB} and using~\eqref{paircondn}, ~\eqref{waveta},~\eqref{wavetb}, we find the validity of the expansion.
\end{proof}

We are now to prove Theorem~\ref{kptau}.
\begin{proof}[Proof of Theorem~\ref{kptau}]
Following Lemma~\ref{lemmaB}, Proposition~\ref{Beq} and $\left[\text{\cite{YZhou},\,Corollary\, 2.1}\right]$(see also~\cite{ABDKS25} Lemma 3.1 or~\cite{Z15} Theorem 5.3), the theorem is proved.
\end{proof}

\begin{remark}
Note that the tau-function for the AKNS hierarchy  was first introduced in~\cite{CZ94} by using the bilinear equations involving two auxiliary variables $\rho,\sigma$ defined by $	q=\rho/\tau^{CZ}, r=\sigma/\tau^{CZ}$.
Indeed, following the results in~\cite{BLX96, C92, FY22, FYZ23} (see also Lemma~\ref{exist-wave}), we can prove that
\begin{align}
	\epsilon^{2}\frac{\partial^{2}\log  \tau^{CZ}}{\partial t_{k}\partial t_{0}}={\rm res}_{\epsilon\partial}L^{k}=2^{k}A_{k}=\epsilon^{2}\frac{\partial^{2}\log  \tau}{\partial t_{k}\partial t_{0}}.
\end{align} 
%Here $L=\epsilon\partial_{X}+q(\epsilon\partial_{X})^{-1}r$ is the scale Lax operator of the AKNS hierarchy. 
This implies that the tau-function $\tau^{CZ}$ agrees with $\tau({\bf t};\epsilon)$. 
\end{remark}

\paragraph{Acknowledgements}
The author is grateful to Professor  Di Yang for his advice, and helpful discussions and suggestions on this note.
The work is partially supported by Scientific Research Foundation for High-level Talents of Anhui University of Science and Technology Z9N4240070.

\smallskip
        
\pdfbookmark[1]{References}{ref}

\medskip
\medskip
\medskip

\medskip

\noindent Ang Fu

\noindent angfu@aust.edu.cn

\medskip

\noindent School of Safety Science and Engineering, Anhui University of Science and Technology, Huainan, 232001, P. R.~China

\noindent School of Mathematical and Big Data, Anhui University of Science and Technology, Huainan, 232001, P.R.~China

\end{document}